\documentclass[12pt]{article}
\usepackage[latin1]{inputenc}
\usepackage{amsmath, graphicx, natbib, layout, verbatim}
\usepackage{setspace, amssymb, amsthm, cases}
\usepackage{fancyhdr, mathrsfs}
\DeclareGraphicsExtensions{.eps}
\usepackage{tikz}
\usetikzlibrary{shapes,backgrounds,arrows,matrix,positioning,chains,fit,calc}

\newcommand{\bG}{\boldsymbol\Gamma}
\newcommand{\bg}{\boldsymbol\gamma}

\newcommand{\bD}{\boldsymbol\Delta}

\newcommand{\bZ}{\mathbf Z}

\newcommand{\bX}{\mathbf X}

\newcommand{\bm}{\boldsymbol m}
\newcommand{\bu}{\boldsymbol u}

\newenvironment{frcseries}{\fontfamily{frc}\selectfont}{}

\newtheorem{prop}{Proposition}
\newtheorem{theo}{Theorem}
\newtheorem{coro}{Corollary}[theo]

\newcommand{\blind}{0}

\begin{document}

\def\spacingset#1{\renewcommand{\baselinestretch}%
{#1}\small\normalsize} \spacingset{1}

\spacingset{1.1}

\if0\blind
{
  \title{\bf Bayesian Estimation of Bipartite Matchings for Record Linkage}
  \author{Mauricio Sadinle\thanks{
Mauricio Sadinle is a Postdoctoral Associate within the Department of Statistical Science, Duke University, Durham, NC 27708 and the National Institute of Statistical Science --- NISS (e-mail: msadinle@stat.duke.edu).  This research is derived from the PhD thesis of the author and was supported by NSF grants SES-11-30706 to Carnegie Mellon University and SES-11-31897 to Duke University/NISS.  The author thanks Kira Bokalders, Bill Eddy, Steve Fienberg, Rebecca Nugent, Jerry Reiter, Beka Steorts, Andrea Tancredi, Bill Winkler, the editors, associate editor and referees for helpful comments and suggestions on earlier versions of this paper, Patrick Ball and Megan Price from the Human Rights Data Analysis Group --- HRDAG for providing access to the data used in this article, and Peter Christen for sharing his synthetic datafile generator.}\hspace{.2cm}\\
    Department of Statistical Science, Duke University, and \\National Institute of Statistical Sciences\\
		}
  \maketitle
} \fi

\if1\blind
{
  \bigskip
  \bigskip
  \bigskip
\title{\bf Bayesian Estimation of Bipartite Matchings for Record Linkage}
\maketitle
  \medskip
} \fi

\begin{abstract}
The bipartite record linkage task consists of merging two disparate datafiles containing information on two overlapping sets of entities.  This is non-trivial in the absence of unique identifiers and it is important for a wide variety of applications given that it needs to be solved whenever we have to combine information from different sources.  Most statistical techniques currently used for record linkage are derived from a seminal paper by Fellegi and Sunter (1969).  These techniques usually assume independence in the matching statuses of record pairs to derive estimation procedures and optimal point estimators.  We argue that this independence assumption is unreasonable and instead target a bipartite matching between the two datafiles as our parameter of interest.  Bayesian implementations allow us to quantify uncertainty on the matching decisions and derive a variety of point estimators using different loss functions.  We propose partial Bayes estimates that allow uncertain parts of the bipartite matching to be left unresolved.  We evaluate our approach to record linkage using a variety of challenging scenarios and show that it outperforms the traditional methodology.  We illustrate the advantages of our methods merging two datafiles on casualties from the civil war of El Salvador.
\end{abstract}
\noindent%
{\it Keywords:} Assignment problem; Bayes estimate; Data matching; Fellegi-Sunter decision rule; Mixture model; Rejection option.

\spacingset{1.2} 

\newpage
\section{Introduction}

Joining data sources requires identifying which entities are simultaneously represented in more than one source.  Although this is a trivial process when unique identifiers of the entities are recorded in the datafiles, in general it has to be solved using the information that the sources have in common on the entities.  Most of the statistical techniques currently used to solve this task are derived from a seminal paper by \cite{FellegiSunter69} who formalized procedures that had been proposed earlier \citep[see][and references therein]{Newcombeetal59,NewcombeKennedy62}.  A number of important record linkage projects have been developed under some variation of the Fellegi-Sunter approach, including the merging of the 1990 U.S. Decennial Census and Post-Enumeration Survey to produce adjusted Census counts \citep{WinklerThibaudeau91}, the Generalized Record Linkage System at Statistics Canada \citep{Fair04}, the Person Identification Validation System at the U.S. Census Bureau \citep{WagnerLayne14}, and the LinkPlus software at the U.S. \cite{CDCLinkPlus}, among many others \citep[e.g.,][]{GillGoldacre03, Singleton13}. 

In this article we are concerned with \emph{bipartite record linkage}, where we seek to merge two datafiles while assuming that each entity is recorded maximum once in each file.  Most of the statistical literature on record linkage deal with this scenario \citep{FellegiSunter69,Jaro89,Winkler88,Winkler93,Winkler94,BelinRubin95,LarsenRubin01,HerzogScheurenWinkler07,TancrediLiseo11,Gutmanetal13}.  Despite the popularity of the Fellegi-Sunter approach and its variants to solve this task, it is also recognized to have a number of caveats \citep[e.g.,][]{Winkler02}.  In particular, the no-duplicates within-file assumption implies a maximum one-to-one restriction in the linkage, that is, a record from one file can be linked with maximum one record from the other file.  Modern implementations of the Fellegi-Sunter methodology that use mixture models ignore this restriction \citep{Winkler88,Jaro89,BelinRubin95,LarsenRubin01}, leading to the necessity of enforcing the maximum one-to-one assignment in a post-processing step \citep{Jaro89}.  Furthermore, this restriction is also ignored by the decision rule proposed by \cite{FellegiSunter69} to classify pairs of records into \emph{links}, \emph{non-links}, and \emph{possible links}, and therefore the conditions for its theoretical optimality are not met in practice.  

Despite the weaknesses of the Fellegi-Sunter approach, it has a number of advantages on which we build in this article, in addition to pushing forward existing Bayesian improvements.  After clearly defining a \emph{bipartite matching} as the parameter of interest in bipartite record linkage (Section \ref{s:general_task}), in Section \ref{s:FS} we review the traditional Fellegi-Sunter methodology, its variants and modern implementations using mixture models, and we provide further details on its caveats.  In Section \ref{s:BBRL} we improve on existing Bayesian record linkage ideas, in particular we extend the modeling approaches of \cite{Fortinietal01} and \cite{Larsen02, Larsen05, Larsen10, Larsen12} to properly deal with missing values and capture partial agreements when comparing pairs of records.  Most importantly, in Section \ref{ss:BRLpointest} we derive Bayes estimates of the bipartite matching according to a general class of loss functions.  Given that Bayesian approaches allow us to properly quantify uncertainty in the matching decisions we include a \emph{rejection option} in our loss functions with the goal of leaving uncertain parts of the bipartite matching undeclared.  The resulting Bayes estimates provide an alternative to the Fellegi-Sunter decision rule.  In Section \ref{s:simulations} we compare our Bayesian approach with the traditional Fellegi-Sunter methodology under a variety of linkage scenarios.  In Section \ref{s:SV} we consider the problem of joining two data sources on civilian casualties from the civil war of El Salvador, and we explain the advantages of using our estimation procedures in that context.  

\section{The Bipartite Record Linkage Task}\label{s:general_task}

Consider two datafiles $\bX_1$ and $\bX_2$ that record information from two overlapping sets of individuals or entities.  These datafiles contain $n_1$ and $n_2$ records, respectively, and without loss of generality we assume $n_1\geq n_2$.  These files originate from two record-generating processes that may induce errors and missing values.  We assume that each individual or entity is recorded maximum once in each datafile, that is, the datafiles contain no duplicates.  Under this set-up the goal of record linkage can be thought of as identifying which records in files $\bX_1$ and $\bX_2$ refer to the same entities.  We denote the number of entities simultaneously recorded in both files by $n_{12}$, and so $0\leq n_{12}\leq n_2$.  Formally, our parameter of interest can be represented by a \emph{bipartite matching} between the two sets of records coming from the two files, as we now explain.

\subsection{A Bipartite Matching as the Parameter of Interest}

We briefly review some basic terminology from graph theory \cite[see, e.g.,][]{LovaszPlummer86}.  A graph $G=(V,E)$ consists of a finite number of elements $V$ called \emph{nodes} and a set of pairs of nodes $E$ called \emph{edges}.  A graph whose node set $V$ can be partitioned into two disjoint non-empty subsets $A$ and $B$ is called \emph{bipartite} if each of its edges connects a node of $A$ with a node of $B$.  A set of edges in a graph $G$ is called a \emph{matching} if all of them are pairwise disjoint.  A matching in a bipartite graph is naturally called a \emph{bipartite matching} (see the example in Figure \ref{f:ToyBipartMatch}).  

\begin{figure*}[t]
\centering
\begin{tikzpicture}[thick,
  every node/.style={draw,circle},
  fsnode/.style={fill=black},
  ssnode/.style={fill=white},
  every fit/.style={rectangle,draw,inner sep=10pt,text width=1.2cm},
  %->,shorten >= 3pt,shorten <= 3pt
]

% the vertices of U
\begin{scope}[start chain=going below,node distance=3mm]
\foreach \i in {1,2,...,5}
  \node[fsnode,on chain] (f\i) [label=left: \i] {};
\end{scope}

% the vertices of V
\begin{scope}[xshift=2.5cm,yshift=0cm,start chain=going below,node distance=3mm]
\foreach \i in {1,2,...,4}
  \node[ssnode,on chain] (s\i) [label=right: \i] {};
\end{scope}

% the set U
\node [black,fit=(f1) (f5),label=above:$A$] {};
% the set V
\node [black,fit=(s1) (s4),label=above:$B$] {};

% the edges
\draw (f1) -- (s2);
\draw (s1) -- (f4);
\draw (f2) -- (s4);
\end{tikzpicture}
  \begin{minipage}[b]{1\textwidth}
  \caption{Example of bipartite matching represented by the edges in this graph.}
\label{f:ToyBipartMatch}\end{minipage} 
\end{figure*}
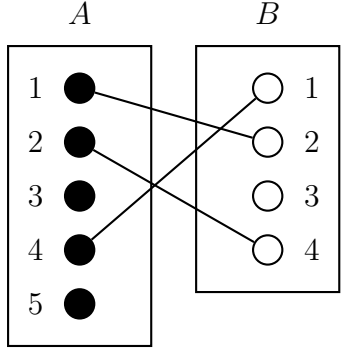

In the bipartite record linkage context we can think of the records from files $\bX_1$ and $\bX_2$ as two disjoint sets of nodes, where an edge between two records represents them referring to the same entity, which we also call being \emph{coreferent} or being a \emph{match}.  The assumption of no duplicates within datafile implies that edges between records of the same file are not possible.  Furthermore, given that the relation of coreference between records is transitive, the graph has to represent a bipartite matching, because if two edges had an overlap, say $(i,j)$ and $(i,j')$, $i\in \bX_1,\ j, j'\in\bX_2$, by transitivity we would have that $j$ and $j'$ would be coreferent, which contradicts the assumption of no within-file duplicates.  

A bipartite matching can be represented in different ways.  The matrix representation consists of creating a \emph{matching matrix} $\bD$ of size $n_1\times n_2$ whose $(i,j)$th entry is defined as
\begin{equation*}
\Delta_{ij}=
\left\{
  \begin{array}{ll}
    1, & \hbox{ if records $i\in \bX_1$ and $j \in \bX_2$ refer to the same entity;} \\
    0, & \hbox{ otherwise.}
  \end{array}
\right.
\end{equation*}
The characteristics of a bipartite matching imply that each column and each row of $\bD$ contain maximum one entry being equal to one.  This representation has been used by a number of authors \citep{LiseoTancredi11, TancrediLiseo11, Fortinietal01, Fortinietal02, Larsen02, Larsen05,Larsen10, Larsen12, Gutmanetal13} but it is not very compact. We propose an alternative way of representing a bipartite matching by introducing a \emph{matching labeling} $\bZ=(Z_{1},Z_{2},\dots,Z_{n_2})$ for the records in the file $\bX_2$ such that 
\begin{equation*}
Z_{j}=
\left\{
  \begin{array}{ll}
    i, & \hbox{ if records $i\in \bX_1$ and $j \in \bX_2$ refer to the same entity;} \\
    n_1+j, & \hbox{ if record $j \in \bX_2$ does not have a match in file $\bX_1$. }
  \end{array}
\right.
\end{equation*}
Naturally we can go from one representation to the other using the relationship $\Delta_{ij}=I(Z_{j}=i)$, where $I(\cdot)$ is the indicator function.  We shall use either representation throughout the document depending on which one is more convenient, although matching labelings are better suited for computations.

\subsection{Approaches to Bipartite Record Linkage}

The goal of bipartite record linkage is to estimate the bipartite matching between two datafiles using the information contained in them.  There are a number of different approaches to do this depending on the specific characteristics of the problem and what information is available.  

A number of approaches directly model the information contained in the datafiles \citep{Fortinietal02, Matsakis10, LiseoTancredi11, TancrediLiseo11,  Gutmanetal13, Steortsetal13}, which requires crafting specific models for each type of field in the datafile, and are therefore currently limited to handle nominal categorical fields, or continuous variables modeled under normality.  In practice, however, fields that are complicated to model, such as names, addresses, phone numbers, or dates, are important to merge datafiles. 

A more common way of tackling this problem is to see it as a traditional classification problem: we need to classify record pairs into matches and non-matches.  If we have access to a sample of record pairs for which the true matching statuses are known, we can train a classifier on this sample using comparisons between the pairs of records as our predictors, and then predict the matching statuses of the remaining record pairs \citep[e.g.,][]{Cochinwalaetal01,Bilenkoetal03,Christen08,Sametal13}.  Nevertheless, classification methods typically assume that we are dealing with i.i.d. data, and therefore the training of the models and the prediction using them heavily rely on this assumption.  In fact, given that these methods output independent matching decisions for pairs of records, they lead to conflicting decisions since they violate the maximum one-to-one assignment constraint of bipartite record linkage.  Typically some subsequent post-processing step is required to solve these inconsistencies.  

Finally, perhaps the most popular approach to record linkage is what we shall call the \emph{Fellegi-Sunter approach}, although many authors have contributed to it over the years.  Despite its difficulties, this approach does not require training data and it can handle any type of field, as long as records can be compared in a meaningful way.  Given that training samples are too expensive to create and datafiles often contain information that is too complicated to model, we believe that the Fellegi-Sunter approach tackles the most common scenarios where  record linkage is needed.  We therefore review this approach in more detail in the next section, and in the remainder of the article we shall refrain from referring to the direct modeling and supervised classification approaches.

\section{The Fellegi-Sunter Approach to Record Linkage}\label{s:FS}

Following \cite{FellegiSunter69}, we can think of the set of ordered record pairs $\bX_1 \times \bX_2$ as the union of the set of \textit{matches} $\mathbf M=\{(i,j); i\in\bX_1,j\in\bX_2, \Delta_{ij}=1\}$ and the set of \textit{non-matches} $\mathbf U=\{(i,j); i\in\bX_1,j\in\bX_2, \Delta_{ij}=0\}$.  The goal when linking two files can be seen as identifying the sets $\mathbf M$ and $\mathbf U$.  When record pairs are estimated to be matches they are called \emph{links} and when estimated to be non-matches they are called \emph{non-links}.  The Fellegi-Sunter approach uses pairwise comparisons of the records to estimate their matching statuses.

\subsection{Comparison Data}

In most record linkage applications two records that refer to the same entity should be very similar, otherwise the amount of error in the datafiles may be too large for the record linkage task to be feasible.  On the other hand, two records that refer to different entities should generally be very different.  Comparison vectors $\bg_{ij}$ are obtained for each record pair $(i,j)$ in $\bX_1 \times \bX_2$ with the goal of finding evidence of whether they represent matches or not.  These vectors can be written as $\bg_{ij}=(\gamma_{ij}^1,\dots,\gamma_{ij}^f,\dots,\gamma_{ij}^F)$, where $F$ denotes the number of criteria used to compare the records.  Traditionally these $F$ criteria correspond to one comparison per each field that the datafiles have in common.  

The appropriate comparison criteria depend on the information contained by the records.  The simplest way to compare the same field of two records is to check whether they agree or not.  This strategy is commonly used to compare unstructured nominal information such as gender or race, but it ignores partial agreements when used with strings or numeric measurements.  To take into account partial agreement among string fields (e.g., names) \cite{Winkler90Strings} proposed to use string metrics, such as the normalized Levenshtein edit distance or any other \citep[see][]{Bilenkoetal03,Elmagarmidetal07}, and divide the resulting set of similarity values into different \emph{levels of disagreement}.  Winkler's approach can be extended to compute levels of disagreement for fields that are not appropriately compared in a dichotomous fashion.

Let $\mathcal{S}_f(i,j)$ denote a similarity measure computed from field $f$ of records $i$ and $j$.  The range of $\mathcal{S}_f$ can be divided into $L_f+1$ intervals $I_{f0}, I_{f1},\dots, I_{fL_f}$, that represent different disagreement levels.  In this construction the interval $I_{f0}$ represents the highest level of agreement, which includes total agreement, and the last interval $I_{fL_f}$ represents the highest level of disagreement, which depending on the field represents complete or strong disagreement.  This allows us to construct the comparison vectors from the ordinal variables:
\begin{equation*}
\gamma^{f}_{ij} = l, \hbox{ if \  } \mathcal{S}_f (i,j) \in I_{fl}.
\end{equation*}
The larger the value of $\gamma^{f}_{ij}$, the more record $i$ and record $j$ disagree in field $f$. 

Although in principle we could define $\bg_{ij}$ using the original similarity values $\mathcal{S}_f(i,j)$, in the Fellegi-Sunter approach these comparison vectors need to be modeled.  Directly modeling the original $\mathcal{S}_f(i,j)$ requires a customized model per type of comparison given that these similarity measures output values in different ranges depending on their functional form and the field being compared.  By building disagreement levels as ordinal categorical variables, however, we can use a generic model for any type of comparison, as long as its values are categorized.

The selection of the thresholds that define the intervals $I_{fl}$ should correspond with what are considered levels of disagreement, which depend on the specific application at hand and the type of field being compared.  For example, in the simulations and applications presented here we build levels of disagreement according to what we consider to be no disagreement, mild disagreement, moderate disagreement, and extreme disagreement.

\subsection{Blocking}

In practice, when the datafiles are large the record linkage task becomes too computationally expensive.  For example, the cost of computing the comparison data alone grows quadratically since there are $n_1\times n_2$ record pairs.  A common solution to this problem is to partition the datafiles into \emph{blocks} of records determined by information that is thought to be accurately recorded in both datafiles, and then solve the task only within blocks.  For example, in census studies datafiles are often partitioned according to ZIP Codes (postal codes) and then only records sharing the same ZIP Code are attempted to be linked, that is, pairs of records with different ZIP Codes are assumed to be non-matches \citep{HerzogScheurenWinkler07}.  Blocking can be used with any record linkage approach and there are different variations \citep[see][for an extensive survey]{Christen12}.  Our presentation in this paper assumes that no blocking is being used, but in practice if blocking is needed the methodologies can simply be applied independently to each block.

\subsection{The Fellegi-Sunter Decision Rule}\label{ss:FS_rule}

The comparison vector $\bg_{ij}$ alone is insufficient to determine whether $(i,j) \in \mathbf M$, since the variables being compared usually contain random errors and missing values.  \cite{FellegiSunter69} used the log-likelihood ratios
\begin{equation}\label{eq:weights}
w_{ij} = \log \frac{\mathbb{P}(\bg_{ij}|\Delta_{ij}=1)}{\mathbb{P}(\bg_{ij}|\Delta_{ij}=0)}
\end{equation}
as weights to estimate which record pairs are matches.  Expression \eqref{eq:weights} assumes that $\bg_{ij}$ is a realization of a random vector, say, $\bG_{ij}$ whose distribution depends on the matching status $\Delta_{ij}$ of the record pair. Intuitively, if this ratio is large we favor the hypothesis of the pair being a match.  Although this type of likelihood ratio was initially used by \cite{Newcombeetal59} and \cite{NewcombeKennedy62}, the formal procedure proposed by \cite{FellegiSunter69} permits finding two thresholds such that the set of weights can be divided into three groups corresponding to links, non-links, and possible links.  The procedure orders the possible values of $\bg_{ij}$ by their weights in non-increasing order, indexing by the subscript $h$, and determines two values, $h'$ and $h''$, such that
\[
\sum_{h\leq h'-1} \mathbb{P}(\bg_{h}|\Delta_{ij}=0)<\mu\leq\sum_{h\leq h'} \mathbb{P}(\bg_{h}|\Delta_{ij}=0)
\]
and
\[
\sum_{h\geq h''} \mathbb{P}(\bg_{h}|\Delta_{ij}=1)\geq\lambda>\sum_{h\geq h''+1} \mathbb{P}(\bg_{h}|\Delta_{ij}=1),
\]
where $\mu=\mathbb{P}(\text{assign } (i,j) \text{ as link}|\Delta_{ij}=0)$ and $\lambda=\mathbb{P}(\text{assign } (i,j) \text{ as non-link}|\Delta_{ij}=1)$ are two admissible error levels.  Finally, the record pairs are divided into three groups: (1) those with
$h\leq h'-1$ being  links, (2) those with $h\geq h''+1$ being non-links, and (3) those with configurations between  $h'$ and $h''$  requiring  clerical review.  \cite{FellegiSunter69} showed that this decision rule is optimal in the sense that for fixed values of $\mu$ and $\lambda$ it minimizes the probability of sending a pair to clerical review.

We notice that in the presence of missing data the sampling distribution of the comparison vectors changes with each missingness pattern, and therefore so do the thresholds $h'$ and $h''$.  The caveats of this decision rule are discussed in Section \ref{ss:FS_caveats}.

\subsection{Enforcing Maximum One-to-One Assignments}\label{ss:Jaro}

The Fellegi-Sunter decision rule does not enforce the maximum one-to-one assignment restriction in bipartite record linkage.  For example, if records $i$ and $i'$ in $\bX_1$ are very similar but are non-coreferent by assumption, and if both are similar to $j\in \bX_2$, then the Fellegi-Sunter decision rule will probably assign $(i,j)$ and $(i',j)$ as links, which by transitivity would imply a link between $i$ and $i'$ (a contradiction).  As a practical solution to this issue, \cite{Jaro89} proposed a tweak to the Fellegi-Sunter methodology.  The idea is to precede the Fellegi-Sunter decision rule with an optimal assignment of record pairs obtained from a linear sum assignment problem.  The problem can be formulated as the maximization:
\begin{align}\label{eq:LSAP}
\max_{\bD} \quad & \sum_{i=1}^{n_1}\sum_{j=1}^{n_2} w_{ij} \Delta_{ij}\\
\text{ subject to }\quad & \Delta_{ij} \in \{0,1\}, \nonumber\\
&\sum_{i=1}^{n_1} \Delta_{ij} \leq 1,\ \ j=1,2,\ldots,n_2, \nonumber\\
&\sum_{j=1}^{n_2} \Delta_{ij} \leq 1,\ \ i=1,2,\ldots,n_1,\nonumber
\end{align}
with $w_{ij}$ given by Expression \eqref{eq:weights}, where the first constraint ensures that $\bD$ represents a discrete structure, and the second and third constraints ensure that each record of $\bX_2$ is matched with at most one record of $\bX_1$ and vice versa.  This is a maximum-weight bipartite matching problem, or a linear sum assignment problem, for which efficient algorithms exist such as the Hungarian algorithm \citep[see, e.g.,][]{PapadimitriouSteiglitz82}.  The output of this step is a bipartite matching that maximizes the sum of the weights among matched pairs, and the pairs that are not matched by this step are considered non-links.  Although \cite{Jaro89} did not provide a theoretical justification for this procedure, we now show that this can be thought of as a maximum likelihood estimate (MLE) under certain conditions, in particular under a conditional independence assumption of the comparison vectors which is commonly used in practice, such as in the mixture models presented in Section \ref{ss:FSMixtureModel}.

\begin{prop} Under the assumption of the comparison vectors being conditionally independent given the bipartite matching, the solution to the linear sum assignment problem in Expression \eqref{eq:LSAP} is the MLE of the bipartite matching.
\end{prop}
\begin{proof} 
\begin{align*}
\hat\bD^{MLE} 
& = \underset{\bD}{\arg\max} \quad \prod_{i,j} \mathbb{P}(\bg_{ij}|\Delta_{ij}=1)^{\Delta_{ij}}\mathbb{P}(\bg_{ij}|\Delta_{ij}=0)^{1-\Delta_{ij}} \\
& = \underset{\bD}{\arg\max} \quad \prod_{i,j} \left[\frac{\mathbb{P}(\bg_{ij}|\Delta_{ij}=1)}{\mathbb{P}(\bg_{ij}|\Delta_{ij}=0)}\right]^{\Delta_{ij}}\\
& = \underset{\bD}{\arg\max} \quad \sum_{i,j} \Delta_{ij}\log\frac{\mathbb{P}(\bg_{ij}|\Delta_{ij}=1)}{\mathbb{P}(\bg_{ij}|\Delta_{ij}=0)},
\end{align*}
where the first line arises under the assumption of the comparison vectors being conditionally independent given the bipartite matching $\bD$, the second line drops a factor that does not depend on $\bD$, and the last line arises from applying the natural logarithm.  We conclude that $\hat\bD^{MLE}$ is the solution to the linear sum assignment problem in Expression \eqref{eq:LSAP}.
\end{proof}
When using $\hat\bD^{MLE}$ there exists the possibility that the matching will include some pairs with a very low matching weight.  \cite{Jaro89} therefore proposed to apply the Fellegi-Sunter decision rule to the pairs that are matched by $\hat\bD^{MLE}$ to determine which of those can actually be declared to be links.

\subsection{Model Estimation}\label{ss:FSMixtureModel}

The presentation thus far relies on the availability of $\mathbb{P}(\cdot|\Delta_{ij}=1)$ and $\mathbb{P}(\cdot|\Delta_{ij}=0)$, but these probabilities need to be estimated in practice.  In principle these distributions could be estimated from previous correctly linked files, but these are seldom available.  As a solution to this problem \cite{Winkler88}, \cite{Jaro89}, \cite{LarsenRubin01}, among others, proposed to model the comparison data using mixture models of the type
\begin{eqnarray}\label{eq:MixtureModel}
\bG_{ij}|\Delta_{ij}=1 & \overset{iid}{\sim}& \mathcal{M}(\bm),\\
\bG_{ij}|\Delta_{ij}=0 & \overset{iid}{\sim}& \mathcal{U}(\bu),\nonumber\\
\Delta_{ij} & \overset{iid}{\sim}& \hbox{Bernoulli}(p),\nonumber
\end{eqnarray}
so that the comparison vector $\bg_{ij}$ is regarded as a realization of a random vector $\bG_{ij}$ whose distribution is either $\mathcal{M}(\bm)$ or $\mathcal{U}(\bu)$ depending on whether the pair is a match or not, respectively, with $\bm$ and $\bu$ representing vectors of parameters and $p$ representing the proportion of matches.  The $\mathcal{M}$ and $\mathcal{U}$ models can be products of individual models for each of the comparison components under a conditional independence assumption \citep{Winkler88,Jaro89}, or can be more complex log-linear models \citep{LarsenRubin01}.  The estimation of these models is usually done using the EM algorithm \citep{DempsterLairdRubin77}.  Notice that the mixture model \eqref{eq:MixtureModel} relies on two key assumptions: the comparison vectors are independent given the bipartite matching, and the matching statuses of the record pairs are independent of one another.  

\subsection{Caveats and Criticism}\label{ss:FS_caveats}

Despite the popularity of the previous methodology for record linkage it has a number of weaknesses.  In terms of modeling the comparison data as a mixture, there is an implicit ``hope'' that the clusters that we obtain are closely associated with matches and non-matches.  In practice, however, the mixture components may not correspond with these groups of record pairs.  In particular, the mixture model will identify two clusters regardless of whether the two files have coreferent records or not.  \cite{Winkler02} mentions conditions for the mixture model to give good results based on experience working with large administrative files at the US Census Bureau: 
\begin{itemize}
\item The proportion of matches should be greater than 5\%.\vspace{-2mm}
\item The classes of matches and non-matches should be well separated.\vspace{-2mm}
\item Typographical error must be relatively low.\vspace{-2mm}
\item There must be redundant fields that overcome errors in other fields.
\end{itemize}
In many practical situations these conditions may not hold, especially when the datafiles contain lots of errors and/or missingness, or when they only have a small number of fields in common.  Furthermore, even if the mixture model is successful at roughly separating matches from non-matches, many-to-one matches can still happen if the assignment step proposed by \cite{Jaro89} is not used, given that the mixture model is fitted without the one-to-one constraint, in particular assuming independence of the matching statuses of the record pairs.  We believe that a more sensible approach is to incorporate this constraint into the model \citep[as in][]{Fortinietal01, Fortinietal02, LiseoTancredi11, TancrediLiseo11, Larsen02, Larsen05,Larsen10, Larsen12,Gutmanetal13} rather than forcing it in a post-processing step.  

Finally, we notice that even if the mixture model is fitted with the one-to-one constraint, the Fellegi-Sunter decision rule alone may still lead to many-to-many assignments and chains of links given that it assumes that once we know the distributions $\mathbb{P}(\cdot|\Delta_{ij}=1)$ and $\mathbb{P}(\cdot|\Delta_{ij}=0)$, the comparison data $\bg_{ij}$ determines the linkage decision.  Furthermore, the optimality of the Fellegi-Sunter decision rule heavily relies on this assumption.  We have argued, however, that the linkage decision for the pair $(i,j)$ not only depends on $\bg_{ij}$ but also depends on the linkage decisions for the other pairs $(i',j)$ and $(i,j')$, $i'\neq i, j'\neq j$.  In Section \ref{ss:BRLpointest} we propose Bayes estimates that allow a rejection option as an alternative to the Fellegi-Sunter decision rule.

\section{A Bayesian Approach to Bipartite Record Linkage}\label{s:BBRL} 

The Bayesian approaches of \cite{Fortinietal01} and \cite{Larsen02, Larsen05, Larsen10, Larsen12} build on the strengths of the Fellegi-Sunter approach but improve on the mixture model implementation by properly treating the parameter of interest as a bipartite matching, therefore avoiding the inconsistencies coming from treating record pairs' matching statuses as independent of one another.  Here we consider an extension of their modeling approach to handle missing data and to take into account multiple levels of partial agreement.  The Bayesian estimation of the bipartite matching (which can be represented by a matching labeling $\bZ$ or by a matching matrix $\bD$) has the advantage of providing a posterior distribution that can be used to derive point estimates and to quantify uncertainty about specific parts of the bipartite matching.  

\subsection{Model for Comparison Data}

Our approach is similar to the mixture model presented in Section \ref{ss:FSMixtureModel}, with the difference that we consider the matching statuses of the record pairs as determined by a bipartite matching:
\begin{eqnarray*}
\bG_{ij}|Z_j=i&\overset{iid}{\sim}&     \mathcal{M}(\bm),\\
\bG_{ij}|Z_j\neq i&\overset{iid}{\sim}& \mathcal{U}(\bu),\nonumber\\
\bZ & \sim & \mathcal{B},\nonumber
\end{eqnarray*}
where $\mathcal{M}(\bm)$ and $\mathcal{U}(\bu)$ are models for the comparison vectors among matches and non-matches, as explained in Section \ref{ss:FSMixtureModel}, and $\mathcal{B}$ represents a prior on the space of bipartite matchings, such as the one presented in Section \ref{ss:BetaMatchingPrior}.  

\subsection{Conditional Independence and Missing Comparisons}\label{ss:CIMAR}

In this section we provide a simple parametrization for the models $\mathcal{M}(\bm)$ and $\mathcal{U}(\bu)$ that allow standard prior specification and make it straightforward to deal with missing comparisons.  Under the assumption of the comparison fields being conditionally independent (CI) given the matching statuses of the record pairs we obtain that the likelihood of the comparison data can be written as 
\begin{align}\label{eq:lhood}
\mathcal{L}(\bZ,\Phi|\bg) = \prod_{i=1}^{n_1}\prod_{j=1}^{n_2}\prod_{f=1}^{F}\prod_{l=0}^{L_f}\Big[
m_{fl}^{I(Z_j=i)}u_{fl}^{I(Z_j\neq i)}\Big]^{I(\gamma^{f}_{ij}=l)},
\end{align}
where $m_{fl}=\mathbb{P}(\Gamma^{f}_{ij}=l|Z_j=i)$ denotes the probability of a match having level $l$ of disagreement in field $f$, and $u_{fl}=\mathbb{P}(\Gamma^{f}_{ij}=l|Z_j\neq i)$ represents the analogous probability for non-matches.  We denote $\bm_f=(m_{f1},\dots,m_{fL_f})$, $\bu_f=(u_{f1},\dots,u_{fL_f})$, $\bm=(\bm_1,\dots,\bm_F)$, $\bu=(\bu_1,\dots,\bu_F)$, and $\Phi=(\bm,\bu)$.  This model is an extension of the one considered by \cite{Larsen02, Larsen05, Larsen10, Larsen12}, which in turn is a parsimonious simplification of the one in \cite{Fortinietal01}, who only considered binary comparisons.

We now need to modify this model to accommodate missing comparison criteria since in practice it is rather common to find records with missing fields of information, which lead in turn to missing comparisons for the corresponding record pairs.  For example, if a certain field that is being used to compute comparison data is missing for record $i$, then the vector $\bg_{ij}$ will be incomplete, regardless of whether the field is missing for record $j$.

A simple way to deal with this situation is to assume that the missing comparisons occur at random
 \citep[MAR assumption in][]{LittleRubin02}, and therefore we can base our inferences on the marginal distribution of the observed comparisons \citep[][p. 90]{LittleRubin02}.  Under the parametrization of Equation \eqref{eq:lhood} and the MAR assumption, after marginalizing over the missing comparisons it can be easily seen that the likelihood of the observed comparison data can be written as 
\begin{align}\label{eq:obs_lhood}
\mathcal{L}(\bZ,\Phi|\bg^{obs}) 
&= \prod_{f=1}^{F}\prod_{l=0}^{L_f} m_{fl}^{a_{fl}(\bZ)}u_{fl}^{b_{fl}(\bZ)},
\end{align}
with 
\begin{align}\label{eq:ab}
a_{fl}(\bZ) &= \sum_{i,j}I_{obs}(\bg_{ij}^f)I(\gamma_{ij}^{f}=l)I(Z_j=i),\\
b_{fl}(\bZ) &= \sum_{i,j}I_{obs}(\bg_{ij}^f)I(\gamma_{ij}^{f}=l)I(Z_j\neq i),\nonumber
\end{align}
where $I_{obs}(\cdot)$ is the indicator of whether its argument is observed.  For a given matching labeling $\bZ$, $a_{fl}(\bZ)$ and $b_{fl}(\bZ)$ represent the number of matches and non-matches with observed disagreement level $l$ in comparison $f$.  From Equations \eqref{eq:obs_lhood} and \eqref{eq:ab} we can see that the combination of the CI and MAR assumptions allow us to ignore the comparisons that are not observed while modeling the observed comparisons in a simple fashion.

Under the previous parametrization it is easy to use the independent conjugate priors $\bm_f\sim \text{Dirichlet} (\alpha_{f0}, \dots, \alpha_{fL_f})$ and $\bu_f\sim \text{Dirichlet}(\beta_{f0},\dots,\beta_{fL_f})$ for $f=1,\dots,F$.

\subsection{Beta Prior for Bipartite Matchings}\label{ss:BetaMatchingPrior}

We now construct a prior distribution for matching labelings $\bZ=(Z_{1},Z_{2},\dots,Z_{n_2})$ where $Z_j\in\{1,2,\dots,n_1,n_1+j\}$ and $Z_j\neq Z_{j'}$.  We start by sampling the indicators of which records in file $\bX_2$ have a match.  Let $I(Z_j\leq n_1)\overset{iid}{\sim}\text{Bernoulli}(\pi)$, $j=1,\dots,n_2$, where $\pi$ represents the proportion of matches expected a priori as a fraction of the smallest file $\bX_2$.  We take $\pi$ to be distributed according to a Beta$(\alpha_\pi,\beta_\pi)$ a priori.  In this formulation $n_{12}(\bZ)=\sum_{j=1}^{n_2}I(Z_j\leq n_1)$ represents the number of matches according to matching labeling $\bZ$, and it is distributed according to a Beta-Binomial$(n_2,\alpha_\pi,\beta_\pi)$, after marginalizing over $\pi$.  Conditioning on knowing which records in file $\bX_2$ have a match, that is, conditioning on $\{I(Z_j\leq n_1)\}_{j=1}^{n_2}$, all the possible bipartite matchings are taken to be equally likely.  There are $n_1!/(n_1-n_{12}(\bZ))!$ such bipartite matchings.  Finally, the probability mass function for $\bZ$ is given by
\begin{equation*}
\mathbb{P}(\bZ|\alpha_\pi,\beta_\pi)=\frac{(n_1-n_{12}(\bZ))!}{n_1!}\frac{\text{B}(n_{12}(\bZ)+\alpha_\pi,n_2-n_{12}(\bZ)+\beta_\pi)}{\text{B}(\alpha_\pi,\beta_\pi)},
\end{equation*}
where B$(\cdot,\cdot)$ represents the Beta function.  We shall refer to this distribution as the \emph{beta distribution for bipartite matchings}.  Notice that in this formulation the hyper-parameters $\alpha_\pi$ and $\beta_\pi$ can be used to incorporate prior information on the amount of overlap between the files.  This prior was first proposed by \cite{Fortinietal01, Fortinietal02} (with fixed $\pi$) and \cite{Larsen05,Larsen10} in terms of matching matrices.

\subsection{Gibbs Sampler}\label{ss:Gibbs}

We now present a Gibbs sampler to explore the joint posterior of $\bZ$  and $\Phi$ given the observed comparison data $\bg_{obs}$, for the likelihood and priors presented before.  Although it is easy to marginalize over $\Phi$ and derive a collapsed Gibbs sampler that iterates only over $\bZ$, we present the expanded version to show some connections with the Fellegi-Sunter approach. 

We start the Gibbs sampler with an empty bipartite matching, that is $Z_j^{[0]}=n_1+j$ for all $j\in \{1,\dots,n_2\}$.  For a current value of the matching labeling $\bZ^{[t]}$, we obtain the next values $\bm_f^{[t+1]}=(m_{f0}^{[t+1]},\dots,m_{fL_f}^{[t+1]})$, $\bu_f^{[t+1]}=(u_{f0}^{[t+1]},\dots,u_{fL_f}^{[t+1]})$, for $f=1,\dots,F$, and $\bZ^{[t+1]}=(Z_1^{[t+1]},\dots,Z_{n_2}^{[t+1]})$ as follows:
\begin{enumerate}
\item For $f=1,\dots,F$, sample 
\begin{align*}
\bm_f^{[t+1]}|\bg^{obs}, \bZ^{[t]} &\sim \text{Dirichlet}(a_{f0}(\bZ^{[t]})+\alpha_{f0},\dots,a_{fL_f}(\bZ^{[t]})+\alpha_{fL_f}),
\end{align*}
and 
\begin{align*}
\bu_f^{[t+1]}|\bg^{obs}, \bZ^{[t]} &\sim \text{Dirichlet}(b_{f0}(\bZ^{[t]})+\beta_{f0},\dots,b_{fL_f}(\bZ^{[t]})+\beta_{fL_f}).
\end{align*}
Collect these new draws into $\Phi^{[t+1]}$.  The functions $a_{fl}(\cdot)$ and $b_{fl}(\cdot)$ are presented in Equation \eqref{eq:ab}.

\item Sample the entries of $\bZ^{[t+1]}$ sequentially.  Having sampled the first $j-1$ entries of $\bZ^{[t+1]}$, we define $\bZ_{-j}^{[t+(j-1)/n_2]}=(Z_1^{[t+1]},\dots,Z_{j-1}^{[t+1]},Z_{j+1}^{[t]},\dots,Z_{n_2}^{[t]})$, and sample a new label $Z_j^{[t+1]}$, with the probability of selecting label $q\in\{1,\dots,n_1,n_1+j\}$ given by $p_{qj}(\bZ_{-j}^{[t+(j-1)/n_2]}, \Phi^{[t+1]})$,
which can be expressed as (for generic $\bZ_{-j}$ and $\Phi$):
\begin{equation}\label{eq:pqjGibbs}
p_{qj}(\bZ_{-j}, \Phi)
\propto 
 \left\{
  \begin{array}{ll}
	 \exp[w_{qj}(\Phi)] I(Z_{j'}\neq q, \forall\ j'\neq j),
			& \hbox{ if $q\leq n_1$;} \\\\
			
    [n_1-n_{12}(\bZ_{-j})] \frac{n_2-n_{12}(\bZ_{-j})-1+\beta_\pi}{n_{12}(\bZ_{-j})+\alpha_\pi}, & \hbox{ if $q=n_1+j$;} 
  \end{array}
\right.
\end{equation}
and  $w_{qj}(\Phi) = \log [ \mathbb{P}(\bg^{obs}_{qj}|Z_j=q,\bm)/\mathbb{P}(\bg^{obs}_{qj}|Z_j\neq q,\bu) ]$ can be expressed as
\begin{align}
w_{qj}(\Phi) &= \sum_{f=1}^F I_{obs}(\bg_{qj}^f)\sum_{l=0}^{L_f}\log\Big(\frac{m_{fl}}{u_{fl}}\Big)I(\gamma_{qj}^{f}=l), \label{eq:WinklerWeight}
\end{align}
\end{enumerate}
for $q\leq n_1$.  From Equations \eqref{eq:pqjGibbs} and \eqref{eq:WinklerWeight} we can see that at a certain step of the Gibbs sampler the assignment of a record $i$ in file $\bX_1$ as a match of record $j$ will depend on the weight $w_{ij}(\Phi^{[t+1]})$, as long as record $i$ does not match any other record of file $\bX_2$ according to $\bZ_{-j}^{[t+(j-1)/n_2]}$.  These are essentially the same weights used in the Fellegi-Sunter approach to record linkage (Section \ref{s:FS}).  In particular, if there are no missing comparisons, Equation \eqref{eq:WinklerWeight} represents the composite weight proposed by \cite{Winkler90Strings} to account for partial agreements.  Equation \eqref{eq:pqjGibbs} also indicates that the probability of not matching record $j$ with any record in file $\bX_1$ depends on the number of unmatched records in file $\bX_1$ and the odds of a non-match in file $\bX_2$ according to $\bZ_{-j}^{[t+(j-1)/n_2]}$.  The lower the number of current matches, the larger the probability of not matching record $j$.

When using a flat prior on the space of bipartite matchings we obtain an expression similar to Equation \eqref{eq:pqjGibbs}, but with a probability of leaving record $j$ unmatched proportional to 1, indicating that under that prior the odds of a match do not take into account the number of existing matches.  In practice this translates into larger numbers of false-matches under the flat prior for scenarios where the actual overlap of the datafiles is small, given that the evidence for a match does not have to be as strong as when using a beta prior for bipartite matchings.  The point estimator $\hat\bD^{MLE}$ presented in Section \ref{ss:Jaro} suffers a similar phenomenon, as we show in Section \ref{s:simulations}.

\section{Bayes Estimates of Bipartite Matchings}\label{ss:BRLpointest}

From a Bayesian theoretical point of view \citep[e.g.,][]{Berger85,BernardoSmith94} we can obtain different point estimates $\hat\bZ$ of the bipartite matching using the posterior distribution of $\bZ$ and different loss functions $L(\bZ,\hat\bZ)$.  The Bayes estimate for $\bZ$ is the bipartite matching $\hat\bZ$ that minimizes the posterior expected loss $\mathbb{E}[L(\bZ,\hat\bZ)|\bg^{obs}]=\sum_{\bZ}L(\bZ,\hat\bZ)\mathbb{P}(\bZ|\bg^{obs})$.  In this section we present a class of additive loss functions that can be used to derive different Bayes estimates.

In some scenarios some records may have a large matching uncertainty, and therefore a point estimate for the whole bipartite matching may not be appropriate.  In Figure \ref{f:UncertainMatch} we present a toy example where a record $j$ in file $\bX_2$ has three possible matches $i,\ i',\ i''$ in file $\bX_1$, making it difficult to take a reliable decision.  The approach presented below allows the possibility of leaving uncertain parts of the bipartite matching unresolved.  Decision rules in the classification literature akin to the ones presented here are said to have a \emph{rejection option}  \citep[see, e.g.,][]{Ripley96,Hu14}.  The rejection option in our context refers to the possibility of not taking a linkage decision for a certain record.  These unresolved cases can, for example, be hand-matched as part of a clerical review.  We refer to point estimates with a rejection option as \emph{partial estimates}, as opposed to \emph{full estimates} which assign a linkage decision to each record.

\begin{figure*}[t]
\centering
		\centerline{\includegraphics[width=0.7\linewidth]{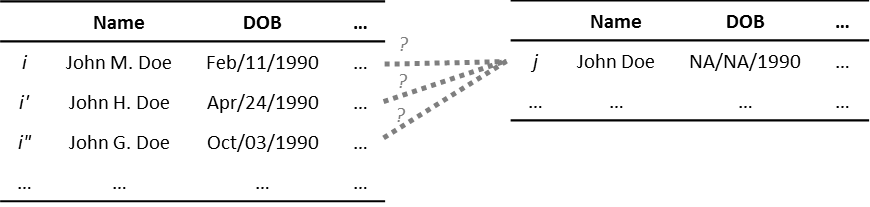}}
  \begin{minipage}[b]{1\textwidth}
  \caption{Toy example of uncertain matching.  DOB: Date of birth.}
\label{f:UncertainMatch}
	\end{minipage} 
\end{figure*}

We work in terms of matching labelings $\bZ$ instead of matching matrices $\bD$, which means that we target questions of the type ``which record in file $\bX_1$ (if any) should be matched with record $j$ in file $\bX_2$?'' rather than ``do records $i$ and $j$ match?''  Working with $\bZ$ makes it explicit that in bipartite record linkage there are $n_2$ linkage decisions to be made rather than $n_1\times n_2$.  

We represent a Bayes estimate here as a vector $\hat\bZ = (\hat Z_{1}, \dots, \hat Z_{n_2})$, where $\hat Z_j\in\{1, \dots, n_1,\\ {n_1+j}, R\}$, with $R$ representing the rejection option.  We propose to assign different positive losses to different types of errors and compute the overall loss additively, as 
\begin{equation}\label{eq:LossR}
L(\bZ,\hat\bZ) = \sum_{j=1}^{n_2} L(Z_j,\hat Z_j),
\end{equation}
with
\begin{equation}\label{eq:LossR_j}
L(Z_j,\hat Z_j)=\left\{
  \begin{array}{ll}
    0, & \hbox{ if } Z_j=\hat Z_j; \\
		\lambda_{R}, & \hbox{ if } \hat Z_j= R; \\
		\lambda_{10}, & \hbox{ if } Z_j\leq n_1, \hat Z_j=n_1+j; \\
		\lambda_{01}, & \hbox{ if } Z_j=n_1+j, \hat Z_j\leq n_1;\\
		\lambda_{11'}, & \hbox{ if } Z_j,\hat Z_j\leq n_1, Z_j\neq \hat Z_j;\\
  \end{array}
\right.
\end{equation}
that is, $\lambda_{R}$ represents the loss from not taking a decision (rejection), $\lambda_{10}$ is the loss from a false non-match decision, $\lambda_{01}$ is the loss from a false match decision when the record does not actually match any other record, and $\lambda_{11'}$ is the loss from a false match decision when the record actually matches a different record than the one assigned to it.  The posterior expected loss is given by
\begin{equation*}\label{eq:PostExpLoss}
\mathbb{E}[L(\bZ,\hat\bZ)|\bg^{obs}] = \sum_{j=1}^{n_2}\varepsilon_{j}(\hat Z_j),
\end{equation*}
where
\begin{equation}\label{eq:PostExpLoss_j}
\varepsilon_{j}(\hat Z_j)=\mathbb{E}[L(Z_j,\hat Z_j)|\bg^{obs}]=\left\{
  \begin{array}{ll}
    \lambda_{R}, & \hbox{ if } \hat Z_j=R; \\		
		\lambda_{10} \mathbb{P}(Z_j\neq n_1+j|\bg^{obs}), & \hbox{ if } \hat Z_j=n_1+j; \\
		\lambda_{01} \mathbb{P}(Z_j= n_1+j|\bg^{obs})+&\\
		\lambda_{11'}\mathbb{P}(Z_j\notin \{ i, n_1+j\}|\bg^{obs}), & \hbox{ if } \hat Z_j=i\leq n_1.\\
  \end{array}
\right.
\end{equation}
The Bayes estimate can be obtained, in general, by solving (minimizing) a linear sum assignment problem with a $(n_1+2n_2)\times n_2$ matrix of weights with entries
\begin{equation*}
v_{ij}=\left\{
  \begin{array}{ll}
		\varepsilon_{j}(i), & \hbox{ if } i\leq n_1; \\		
		\varepsilon_{j}(n_1+j), & \hbox{ if } i=n_1+j; \\
		\lambda_{R}, & \hbox{ if } i=2n_1+j;\\
		\infty, & \hbox{otherwise}.
  \end{array}
\right.
\end{equation*}
In this matrix the first $n_1$ rows accommodate the possibility of records in file $\bX_2$ linking to any record in file $\bX_1$, the next $n_2$ rows accommodate the possibility of records in $\bX_2$ not linking to any record in $\bX_1$, and the last $n_2$ rows represent the possibility of not taking linkage decisions (rejections) for the records in $\bX_2$.  Rather than working with this general formulation we now focus on  some important particular cases that lead to simple derivations of the Bayes estimates.  

\subsection{Closed-Form Full Estimates with Conservative Link Assignments}\label{ss:Zhat}

We first consider the case where we are required to output decisions for all records.  In this case fixing $\lambda_{R}=\infty$ prevents outputting rejections.  Letting $\lambda_{10}\leq \lambda_{01}, \lambda_{11'}$ represents the idea that the loss from a false non-match is not higher than the possible losses from a false match.  Furthermore, the error of matching $j$ with a record $i$ when it actually matches another $i'\neq i$ implies that record $i'$ will not be matched correctly either, and therefore it is reasonable to take $\lambda_{11'}$ to be much larger than the other losses.  In particular we work with $\lambda_{11'}\geq \lambda_{10} + \lambda_{01}$.

\begin{theo}  \label{prop:full_link_gen}
If $\lambda_{R}=\infty, 0<\lambda_{10}\leq \lambda_{01}$, and $\lambda_{11'}\geq \lambda_{10} + \lambda_{01}$ in the loss function given by Equations \eqref{eq:LossR} and \eqref{eq:LossR_j}, the Bayes estimate of the bipartite matching is obtained from $\hat\bZ = (\hat Z_{1}, \dots, \hat Z_{n_2})$, where $\hat Z_j$ is given by
\begin{equation*}\label{eq:hatZjF}
\hat Z_j=\left\{
  \begin{array}{ll}
    i, & \hbox{ if } \mathbb{P}(Z_j=i|\bg^{obs})> \frac{\lambda_{01}}{\lambda_{01}+\lambda_{10}}+\frac{\lambda_{11'}-\lambda_{01}-\lambda_{10}}{\lambda_{01}+\lambda_{10}}\mathbb{P}(Z_j\notin\{ i, n_1+j\}|\bg^{obs}); \\
		n_1+j, & \hbox{ otherwise. } 
  \end{array}
\right.
\end{equation*}
\end{theo}

\begin{proof}
The strategy for the proof is to obtain the optimal marginal value of each $\hat Z_j$ by minimizing each term $\varepsilon_{j}(\hat Z_j)$ shown in Equation \eqref{eq:PostExpLoss_j}.  If this approach leads to a proper bipartite matching then it corresponds to the optimal solution of the problem given that the constraints $\hat Z_j\neq \hat Z_{j'}$ for $j\neq j'$ would hold.  

To find the optimal value of $\hat Z_j$ we can start by finding the optimal label among $\{1,\dots,n_1\}$.  It is easy to see that $i^*$  minimizes $\varepsilon_{j}(i)$ if and only if it maximizes $\mathbb{P}(Z_j=i|\bg^{obs})$.  Now, if $i^*$ is the best possible match for $j$, the decision of matching $j$ with $i^*$ over not matching $j$ with any other record depends on whether $\varepsilon_{j}(n_1+j)>\varepsilon_{j}(i^*)$,  which can easily be checked to be equivalent to the inequality stated in the theorem.  

Given that this solution was obtained ignoring the constraints that require $\hat Z_j\neq \hat Z_{j'}$ for $j\neq j'$ we need to make sure that it leads to a bipartite matching.  Indeed, given the conditions on  $\lambda_{10}, \lambda_{01}$ and $\lambda_{11'}$ we have $\lambda_{01}/(\lambda_{01}+\lambda_{10})\geq 1/2$ and $(\lambda_{11'}-\lambda_{01}-\lambda_{10})/(\lambda_{01}+\lambda_{10})\geq 0$, which imply that under this solution $\hat Z_{j}=i$ only if $\mathbb{P}(Z_{j}=i|\bg^{obs})>1/2$.  Since we are working with a posterior distribution on bipartite matchings we necessarily have that $\sum_{j=1}^{n_2}\mathbb{P}(Z_{j}=i|\bg^{obs})\leq 1$, given that the events $Z_1=i,Z_2=i,\dots,Z_{n_2}=i$ are disjoint.  This implies that $\mathbb{P}(Z_{j'}=i|\bg^{obs})< 1/2$ for all $j'\neq j$, and so $\hat Z_{j'}\neq i$ for all $j'\neq j$.  We conclude that the solution given by the theorem satisfies the constrained problem. 
\end{proof}

The conservative nature of the Bayes estimates obtained from Theorem \ref{prop:full_link_gen} are evidenced from the fact that to declare a match between records $j$ and $i$ we require $\mathbb{P}(Z_j=i|\bg^{obs})$ to be at least $\lambda_{01}/(\lambda_{01}+\lambda_{10})\geq 1/2$.  Furthermore, in cases where record $j$ has a non-zero probability of matching other records besides $i$, that is, when $\mathbb{P}(Z_j\notin\{ i, n_1+j\}|\bg^{obs})>0$, if $\lambda_{11'}> \lambda_{10} + \lambda_{01}$ the decision rule in Theorem \ref{prop:full_link_gen} is extra  conservative increasing the threshold $\lambda_{01}/(\lambda_{01}+\lambda_{10})$ for declaring matches. 

The Bayes estimate of Theorem \ref{prop:full_link_gen} has an important particular case.  \cite{TancrediLiseo11} derived a decision rule using the entrywise zero-one loss for matching matrices
\begin{equation*}
L_{e01}(\bD,\hat \bD)=\sum_{i=1}^{n_1}\sum_{j=1}^{n_2} I(\Delta_{ij}\neq\hat\Delta_{ij}),
\end{equation*}
which is equivalent to our additive loss function when $\lambda_{01}=\lambda_{10}=1$, $\lambda_{11'}=2$ in Equations \eqref{eq:LossR} and \eqref{eq:LossR_j}, and therefore we obtain the following corollary.
\begin{coro}[\cite{TancrediLiseo11}]
If $\lambda_{R}=\infty, \lambda_{10}= \lambda_{01}=1$, and $\lambda_{11'}=2$ in the loss function given by Equations \eqref{eq:LossR} and \eqref{eq:LossR_j}, the Bayes estimate of the bipartite matching is obtained from $\hat\bZ^{e01} = (\hat Z^{e01}_{1}, \dots, \hat Z^{e01}_{n_2})$, where $\hat Z^{e01}_j$ is given by
\begin{equation}\label{eq:Ze01}
\hat Z_j^{e01}=\left\{
  \begin{array}{ll}
    i, & \hbox{ if } \mathbb{P}(Z_j=i|\bg^{obs})> 1/2; \\
		n_1+j, & \hbox{ otherwise. } 
  \end{array}
\right.
\end{equation}
\end{coro}

\subsection{Closed-Form Partial Estimates}\label{ss:Zhat_rejections}

To emphasize the importance of the rejection option, let us refer to the toy example of Figure \ref{f:UncertainMatch}, where a record $j$ in file $\bX_2$ has three possible matches $i,\ i',\ i''$ in file $\bX_1$.  If each of these matches is equally likely we necessarily have that $\mathbb{P}(Z_{j}=i|\bg^{obs})< 1/2$, and likewise for $i'$ and $i''$.  In this case the optimal decision under the  entrywise zero-one loss is to not match $j$ with any record in file $\bX_1$.  On the other hand, in the case of the bipartite matching MLE $\hat\bD^{MLE}$ (Section \ref{ss:Jaro}), if the three weights $w_{ij}$, $w_{i'j}$, $w_{i''j}$ are equal and positive then this estimate will arbitrarily match $j$ with one of $i$, $i'$, or $i''$ (if no other records are involved).  This scenario illustrates the advantage of using a decision rule that allows us to leave uncertain parts of the bipartite matching unresolved.  

We now present a particular case of our additive loss function that allows us to output rejections and leads to closed-form Bayes estimates.  At the end of this section we explain why the constraints that we consider on the individual losses are meaningful in practice.  

\begin{theo}\label{prop:partial_link_gen}
If either 1) $\lambda_{11'}\geq \lambda_{01} \geq 2\lambda_R>0$, or 2) $\lambda_{01}\geq \lambda_{10}>0$ and $\lambda_{11'}\geq \lambda_{01}+\lambda_{10}$, in the loss function given by Equations \eqref{eq:LossR} and \eqref{eq:LossR_j}, the Bayes estimate of the bipartite matching can be obtained from $\hat\bZ = (\hat Z_{1}, \dots, \hat Z_{n_2})$, with $\hat Z_j={\arg\min}_{\hat Z_j}\ \varepsilon_{j}(\hat Z_j)$, where $\varepsilon_{j}(\hat Z_j)$ is given by Expression \eqref{eq:PostExpLoss_j}.
\end{theo}
\begin{proof}
We need to show that $\hat\bZ$ is such that if $\hat Z_j\in\{1,\dots,n_1\}$ then $\hat Z_{j'}\neq \hat Z_j$ for $j'\neq j$, that is, we do not obtain conflicting matching decisions.  
\begin{enumerate}
\item[1)] Assume $\lambda_{11'}\geq \lambda_{01} \geq 2\lambda_R>0$. According to the construction of $\hat\bZ$ in the theorem, if $\hat Z_j=i\in \{1,\dots,n_1\}$ then $\varepsilon_{j}(i)<\varepsilon_{j}(R)$, which is equivalent to 
\begin{align*}
\mathbb{P}(Z_j=i|\bg^{obs}) & > 1-\frac{\lambda_{R}}{\lambda_{01}}+\frac{\lambda_{11'}-\lambda_{01}}{\lambda_{01}}\mathbb{P}(Z_j\notin\{ i, n_1+j\}|\bg^{obs}).
\end{align*}
Using this inequality along with the restrictions $\lambda_{11'}\geq \lambda_{01} \geq 2\lambda_R>0$ we obtain $\mathbb{P}(Z_j=i|\bg^{obs})>1/2$, which implies that $\mathbb{P}(Z_{j'}=i|\bg^{obs})<1/2$ because $\sum_{j=1}^{n_2}\mathbb{P}(Z_{j}=i|\bg^{obs})\leq 1$, and therefore $\hat Z_{j'}\neq \hat Z_{j}$ for all $j'\neq j$.  
\item[2)] Assume $\lambda_{01}\geq \lambda_{10}>0$ and $\lambda_{11'}\geq \lambda_{01}+\lambda_{10}$.  If $\hat Z_j=i\in \{1,\dots,n_1\}$ then $\varepsilon_{j}(i)<\varepsilon_{j}(n_1+j)$.  In the proof of Theorem \ref{prop:full_link_gen} we showed that in such case if $\lambda_{01}\geq \lambda_{10}$ and $\lambda_{11'}\geq \lambda_{01}+\lambda_{10}$ then $\mathbb{P}(Z_{j}=i|\bg^{obs})>1/2$, which in turn implies that $\mathbb{P}(Z_{j'}=i|\bg^{obs})< 1/2$ for all $j'\neq j$, and so $\hat Z_{j'}\neq i$ for all $j'\neq j$.
\end{enumerate}
\end{proof}

We now present a particular case of Theorem \ref{prop:partial_link_gen} that allows an explicit expression for the Bayes estimate.  

\begin{theo}\label{prop:partial_link_closed_form} If $\lambda_{11'}\geq \lambda_{01} \geq 2\lambda_R> 0$ and $\lambda_{10}\geq 2\lambda_{R}$ in the loss function given by Equations \eqref{eq:LossR} and \eqref{eq:LossR_j}, the Bayes estimate of the bipartite matching can be obtained from $\hat\bZ = (\hat Z_{1}, \dots, \hat Z_{n_2})$, where $\hat Z_j$ is given by
\begin{subnumcases}{\hat Z_j=}
    i, &  if \ $\mathbb{P}(Z_j=i|\bg^{obs})>1-\frac{\lambda_{R}}{\lambda_{01}}+\frac{\lambda_{11'}-\lambda_{01}}{\lambda_{01}}\mathbb{P}(Z_j\notin\{ i, n_1+j\}|\bg^{obs})$; \label{eq:colo2_ineq1}\\
		n_1+j, &  if \ $\mathbb{P}(Z_j=n_1+j|\bg^{obs})>1-\frac{\lambda_R}{\lambda_{10}}; \label{eq:colo2_ineq2}$\\
		R, &  otherwise.\nonumber
\end{subnumcases}
\end{theo}
\begin{proof}
From Theorem \ref{prop:partial_link_gen} we know that the constraints $\lambda_{11'}\geq \lambda_{01} \geq 2\lambda_R> 0$ allow the Bayes estimate to be obtained from the marginal optimal decisions for each $\hat Z_j$.  We now need to show: 1) the inequality in \eqref{eq:colo2_ineq1} holds true if and only if $\varepsilon_{j}(i)<\varepsilon_{j}(R),\ \varepsilon_{j}(n_1+j)$; 2) the inequality in \eqref{eq:colo2_ineq2} holds true if and only if $\varepsilon_{j}(n_1+j)<\varepsilon_{j}(R),\ \varepsilon_{j}(i)$ for all $i\in\{1,\dots,n_1\}$.  
\begin{itemize}
\item[1)] Firstly, as it had been noted in the proof of Theorem \ref{prop:partial_link_gen}, the inequality in \eqref{eq:colo2_ineq1} is equivalent to $\varepsilon_{j}(i)<\varepsilon_{j}(R)$.  ($\Rightarrow$) Given the previous note we only need to show $\varepsilon_{j}(i)<\varepsilon_{j}(n_1+j)$.  Given the constraints $\lambda_{11'}\geq \lambda_{01} \geq 2\lambda_R> 0$, the inequality in \eqref{eq:colo2_ineq1} implies $\mathbb{P}(Z_j=i|\bg^{obs})>1/2$, which in turn implies $\mathbb{P}(Z_j=n_1+j|\bg^{obs})<1/2\leq 1-\lambda_{R}/\lambda_{10}$ (because $\lambda_{10}\geq 2\lambda_{R}$), which is equivalent to $\varepsilon_{j}(R)<\varepsilon_{j}(n_1+j)$.  By transitivity we have $\varepsilon_{j}(i)<\varepsilon_{j}(n_1+j)$.  ($\Leftarrow$)  If $i={\arg\min}_{\hat Z_j}\ \varepsilon_{j}(\hat Z_j)$, then in particular $\varepsilon_{j}(i)<\varepsilon_{j}(R)$, which is equivalent to the inequality in \eqref{eq:colo2_ineq1}.
\item[2)] The inequality in \eqref{eq:colo2_ineq2} is equivalent to $\varepsilon_{j}(n_1+j)<\varepsilon_{j}(R)$.  ($\Rightarrow$) We only need to show that $\varepsilon_{j}(n_1+j)<\varepsilon_{j}(i)$ for all $i\in\{1,\dots,n_1\}$.  If the inequality in \eqref{eq:colo2_ineq2} holds true then $\mathbb{P}(Z_j=n_1+j|\bg^{obs})>1/2$, which implies $\mathbb{P}(Z_j=i|\bg^{obs})<1/2$, which in turn means that the inequality in \eqref{eq:colo2_ineq1} does not hold true for any $i$, and therefore $\varepsilon_{j}(R)\leq\varepsilon_{j}(i)$ for all $i\in\{1,\dots,n_1\}$ (because inequality \eqref{eq:colo2_ineq1} is equivalent to $\varepsilon_{j}(i)<\varepsilon_{j}(R)$). The result is obtained by transitivity.  ($\Leftarrow$) If $n_1+j={\arg\min}_{\hat Z_j}\ \varepsilon_{j}(\hat Z_j)$, then in particular $\varepsilon_{j}(n_1+j)<\varepsilon_{j}(R)$, which is equivalent to the inequality in \eqref{eq:colo2_ineq2}.
\end{itemize}
\end{proof}

Under the conditions of Theorem \ref{prop:partial_link_closed_form}, if the only two probable possibilities for record $j$ are to either match a certain record $i$ or to not match any record, that is, if $\mathbb{P}(Z_j\notin\{ i, n_1+j\}|\bg^{obs})=0$, or if the loss from a false match between records $j$ and $i$ is the same regardless of the actual matching status of $j$, that is, if $\lambda_{11'}=\lambda_{01}$, then we take the decision $\hat Z_j=i$ only when $\mathbb{P}(Z_j\neq i|\bg^{obs})<\lambda_{R}/\lambda_{01}$, and so for such cases $\lambda_{R}/\lambda_{01}$ works as a control over the probability of a false match.  It is therefore sensible to take $\lambda_{R}/\lambda_{01}$ to be small, and in particular the solution in Theorem \ref{prop:partial_link_closed_form} covers the cases when $\lambda_{R}/\lambda_{01}\leq 1/2$.  For cases when $\mathbb{P}(Z_j\notin\{ i, n_1+j\}|\bg^{obs})>0$, if $\lambda_{11'}>\lambda_{01}$, the decision rule in Theorem \ref{prop:partial_link_closed_form} is more conservative requiring $\mathbb{P}(Z_j\neq i|\bg^{obs})$ to be even lower than $\lambda_{R}/\lambda_{01}$ to declare a match.  Finally, under the conditions of Theorem \ref{prop:partial_link_closed_form} we take the decision $\hat Z_j=n_1+j$ only when $\mathbb{P}(Z_j\neq n_1+j|\bg^{obs})<\lambda_{R}/\lambda_{10}$.  Given that $\lambda_{R}/\lambda_{10}$ works effectively as a control over the probability of a false non-match, only small values of $\lambda_{R}/\lambda_{10}$ are sensible, and the solution in Theorem \ref{prop:partial_link_closed_form} covers such cases since it requires $\lambda_{R}/\lambda_{10}\leq 1/2$.

\section{Performance Comparison}\label{s:simulations}

We now present a simulation study to compare the performance of the Fellegi-Sunter mixture model approach with the approach presented in Section \ref{s:BBRL}, which for simplicity we refer to as \emph{beta record linkage}, given that we use the beta prior for bipartite matchings (Section \ref{ss:BetaMatchingPrior}).  We consider  different scenarios of files' overlap and measurement error.  We generated pairs of datafiles using a synthetic data generator developed by \cite{Christen05}, \cite{ChristenPudjijono09}, and \cite{ChristenVatsalan13}.  This tool allows us to create synthetic datasets containing various types of fields which can be corrupted with different types of errors.  Since it would be expected for a record linkage methodology to perform well when the records have a lot of identifying information, we are interested in a more challenging scenario where decisions have to be made based on only a small number of fields.  

In this simulation each datafile has 500 records and four fields: given and family names, age, and occupation.  For each pair of datafiles there are $n_{12}$ individuals included in both, and so we refer to them as their \textit{overlap}.  We generated 100 pairs of datafiles for each combination of 100\%, 50\%, and 10\% files' overlap, and 1, 2, and 3 erroneous fields per record.  To generate each pair of datafiles the fields given and family names are sampled from frequency tables compiled by Christen et al. from public sources in Australia, and therefore popular names appear with higher probability in the synthetic datasets.  Age and occupation are each represented by eight categories and are sampled from their joint distribution in Australia.  
The data generator first creates a number of clean records which are later distorted to create the datafiles.  Each distorted record has a fixed number of erroneous fields which are allocated uniformly at random, and each field contains a maximum of three errors.  The types of errors are selected uniformly at random from a set of possibilities which vary from field to field, as summarized in Table \ref{t:errors_simulation}.  Notice that we generate missing values only for the fields age and occupation.  

For each pair of files we computed comparison data as summarized in Table \ref{t:sim_compdata}.  To compare names we use the Levenshtein edit distance, which is the minimum number of deletions, insertions, or replacements that we need to transform one string into the other.  We standardize this distance by dividing it by the length of the longest string.  The final measure belongs to the unit interval with 0 and 1 representing total agreement and total disagreement, respectively.

\begin{table}[t]
\centering
\footnotesize{
  \begin{minipage}[b]{.95\textwidth}
 \caption{Types of errors per field in the simulation study.  Edits: insertions, deletions, or substitutions of characters in a string.  OCR: optical character recognition errors. Keyboard: typing errors that rely on a certain keyboard layout.  Phonetic: uses a list of predefined phonetic rules.  For further details on the generation of these types of errors see \cite{ChristenPudjijono09} and \cite{ChristenVatsalan13}.}\label{t:errors_simulation}
\centering
\begin{tabular}{lccccc}
  \hline\\[-8pt]
       &  \multicolumn{5}{c}{Type of Error}\\
          \cline{2-6}\\ [-6pt]
  Fields & Missing & Edits & OCR & Keyboard & Phonetic \\
  \hline\\[-8pt]	
	Given and Family Names & & $\checkmark$ & $\checkmark$ & $\checkmark$ & $\checkmark$ \\  
	Age and Occupation & $\checkmark$ & && $\checkmark$ & \\
  \hline
\end{tabular}
\end{minipage}
}
\end{table}

\begin{table}[t]
\centering
\footnotesize{
  \begin{minipage}[b]{.95\textwidth}
 \caption{Construction of disagreement levels in the simulation study.  The Levenshtein distance is standardized to be in the unit interval.}\label{t:sim_compdata}
\centering
\begin{tabular}{lccccc}
  \hline\\[-8pt]
       &                     & \multicolumn{4}{c}{Levels of Disagreement}\\
          \cline{3-6}\\ [-6pt]
  Fields & Similarity & $0$ & $1$ & $2$ & $3$ \\
  \hline\\[-8pt]
  Given and Family Names & Levenshtein &  0 & $(0,.25]$ & $(.25,.5]$ & $(.5,1]$ \\
  Age and Occupation & Binary & Agree & Disagree &&\\
  \hline
\end{tabular}
\end{minipage}
}
\end{table}

We implemented the Fellegi-Sunter mixture model approach using the same likelihood that we used for beta record linkage (Equation \eqref{eq:obs_lhood}) and the EM algorithm.  In our Bayesian approach we used flat priors on the $\bm_f$ and $\bu_f$ parameters for all $f$, and also on the proportion of matches $\pi$ (see Section \ref{ss:BetaMatchingPrior}), that is $\alpha_\pi=\beta_{\pi}=\alpha_{f0}=\dots=\alpha_{fL_f}=\beta_{f0}=\dots=\beta_{fL_f}=1$ for all comparisons $f=1,2,3,4$.  For each pair of datasets we ran 1,000 iterations of the Gibbs sampler presented in Section \ref{ss:Gibbs}, and discarded the first 100 as burn-in.  The average runtime using an implementation in R \citep{R13} with parts written in C language was of 22, 32, and 37 seconds for files with overlap 100\%, 50\%, and 10\%, respectively, including the computation of the comparison data, on a laptop with a 2.80 GHz processor.  The corresponding average runtimes for the Fellegi-Sunter approach using an R implementation were 11, 18, and 54 seconds per file.  Although the software implementations of both methodologies are not comparable, they indicate that the Fellegi-Sunter mixture model approach can be much faster.  We also implemented a Bayesian alternative to the beta approach using a flat prior on the bipartite matchings, but its performance is virtually the same as the Fellegi-Sunter approach, and so we do not report these results.  

\subsection{Results with Full Estimates}

For each pair of files we obtain a full point estimate of the bipartite matching using each approach.  For the Fellegi-Sunter approach we use the MLE of the bipartite matching (Section \ref{ss:Jaro}) and for beta record linkage we use the $\hat\bZ^{e01}$ estimate obtained from Equation \eqref{eq:Ze01}.  For each estimate we computed the measures of \emph{precision} and \emph{recall}.  
If $\bZ$ is the true bipartite matching labeling, then the recall of an estimate $\hat\bZ$ is the proportion of matches that are correctly linked by $\hat\bZ$, that is $\sum_{j=1}^{n_2}I(\hat Z_j=Z_j\leq n_1)/\sum_{j=1}^{n_2}I(Z_j\leq n_1)$, whereas the precision of $\hat\bZ$ is the proportion of records linked by $\hat\bZ$ that are actual matches, that is $\sum_{j=1}^{n_2}I(\hat Z_j=Z_j\leq n_1)/\sum_{j=1}^{n_2}I(\hat Z_j\leq n_1)$.  A perfect record linkage procedure would lead to precision = recall = 1.
To summarize the performance of the methods under each scenario of overlap and measurement error we computed the median, first, and 99th percentiles of these measures across the 100 pairs of datafiles.

In Figure \ref{f:BRLAustraliaSim} we present the results of the simulation study, where rows show the performance of the two approaches and columns show the results for different amounts of overlap between the files.  In each panel solid lines refer to precision and dashed lines to recall, black lines show medians and gray lines show first and 99th percentiles.  We can see from the first row of Figure \ref{f:BRLAustraliaSim} that the Fellegi-Sunter mixture model approach has an excellent performance when the files have a large overlap, but its precision decays when the overlap of the files decrease, meaning that this methodology generates a large proportion of false-matches.  These findings agree with the observations made by \cite{Winkler02} in the sense that the mixture model approach leads to poor results when the overlap of the files is small and when the files do not contain a lot of identifying information.  Under these scenarios the mixture model is not able to accurately identify the clusters of record pairs associated with matches and non-matches, and instead separates a cluster of extreme disagreement profiles from the rest, leading to a large number of false-matches.

\begin{figure*}[t]
\centering
		\centerline{\includegraphics[width=.85\linewidth]{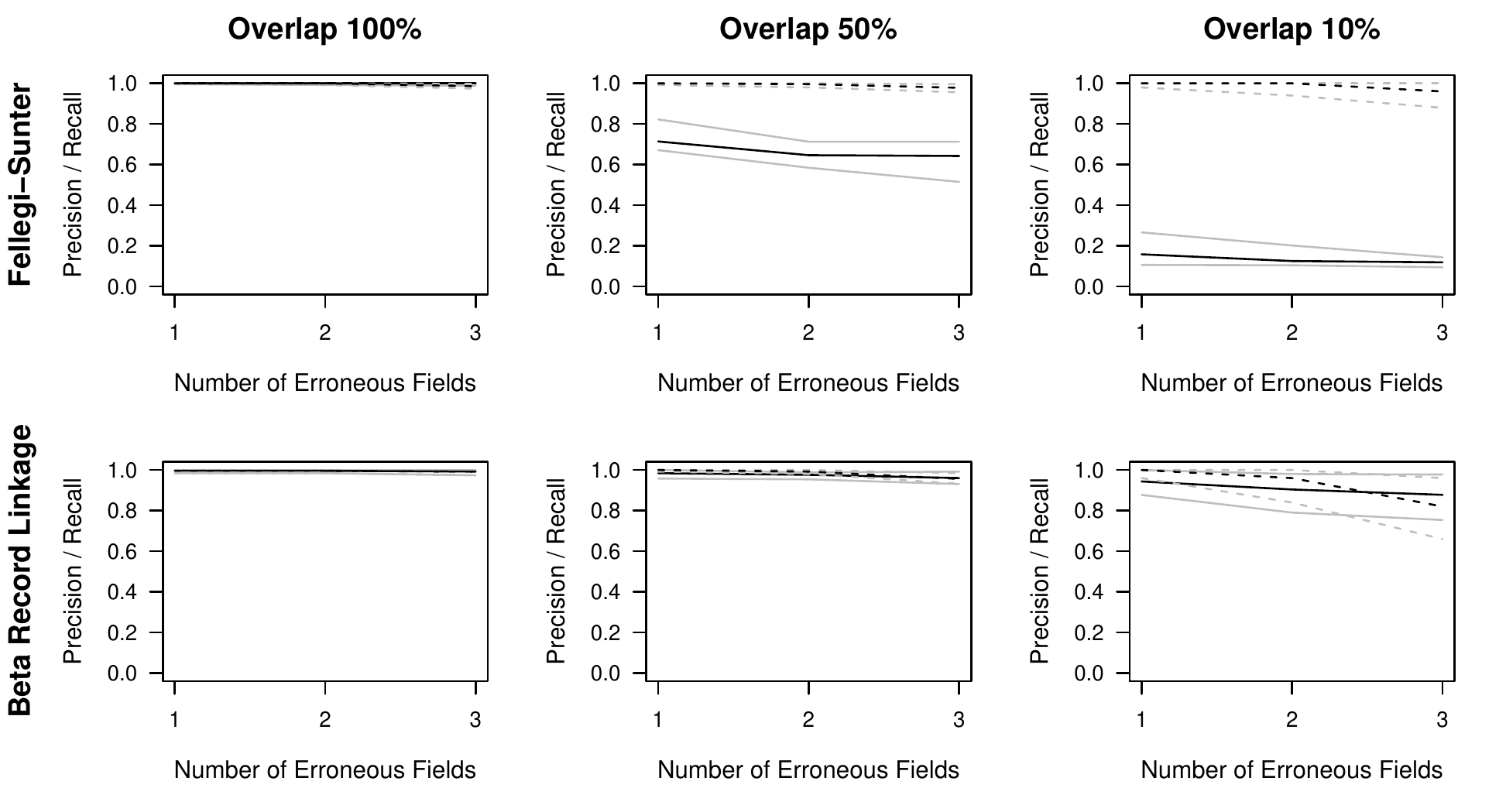}}
  \begin{minipage}[b]{.95\textwidth}
  \caption{Comparison of the performance of two methodologies for record linkage.  Solid lines refer to precision, dashed lines to recall, black lines show medians, and gray lines show first and 99th percentiles.}
\label{f:BRLAustraliaSim}\end{minipage} 
\end{figure*}

On the other hand, from the second row of Figure \ref{f:BRLAustraliaSim} we can see that the performance of the beta approach is remarkable  across all scenarios, and even though it deteriorates when the number of errors increases and when the overlap of the files decreases, it is much more robust than the Fellegi-Sunter approach.  In scenarios where the amount of error is large and the overlap is small, the  uncertainty in the linkage may be quite large, which is evidenced by the variability of the results in the panel of the second row and third column of Figure \ref{f:BRLAustraliaSim}.  For such cases it can be beneficial to leave uncertain parts of the bipartite matching undecided.  We now show the performance of the beta approach when allowing a rejection option as introduced in Section \ref{ss:Zhat_rejections}.

\subsection{Results with Partial Estimates}

We now present the performance of both methodologies when allowing for a rejection option.  In the case of the Fellegi-Sunter approach this is done by using the Fellegi-Sunter decision rule presented in Section \ref{ss:FS_rule} after obtaining the MLE of the bipartite matching.  We use the nominal error levels $\mu=\mathbb{P}(\text{assign } (i,j) \text{ as link}|\Delta_{ij}=0)= 0.0025$ and $\lambda=\mathbb{P}(\text{assign } (i,j) \text{ as non-link}|\Delta_{ij}=1)=0.005$.  We fix the nominal error levels as $\mu=\lambda/2$ to (nominally) protect more against false-matches.  

For beta record linkage we use the Bayes estimate presented in Theorem \ref{prop:partial_link_closed_form} with $\lambda_{10}=\lambda_{01}=1$ and $\lambda_{11'}=2$, so that it is equivalent to adding a rejection option to the $\hat\bZ^{e01}$ estimator used in the previous section.  We fix the loss of a rejection at $\lambda_{R}=0.1$ so that a rejection is only 10\% as costly as a false non-match.  We notice that this partial estimator is not comparable with the Fellegi-Sunter decision rule in terms of aiming to control the nominal errors $\mu$ and $\lambda$.  

When we allow rejections it is not meaningful to use the measure of recall anymore given that we are not aiming at detecting all the matches.  Instead, we are now aiming at being accurate with the decisions that we take.  In this section we therefore use two measures of accuracy for our linkage and non-linkage decisions.  The \emph{positive predictive value} (PPV) is the proportion of links that are actual matches, and so it is equivalent to the precision measure used in the last section.  The \emph{negative predictive value} (NPV) is defined as the proportion of non-links that are actual non-matches, that is $\sum_{j=1}^{n_2}I(\hat Z_j = Z_j = n_1+j)/\sum_{j=1}^{n_2}I(\hat Z_j = n_1+j)$.  In addition we report the \emph{rejection rate} (RR), defined as $\sum_{j=1}^{n_2}I(\hat Z_j = R )/n_2$, which should ideally be small.  A perfect record linkage procedure would have PPV $=$ NPV $= 1$ and RR $=0$.  

\begin{figure*}[t]
\centering
		\centerline{\includegraphics[width=0.9\linewidth]{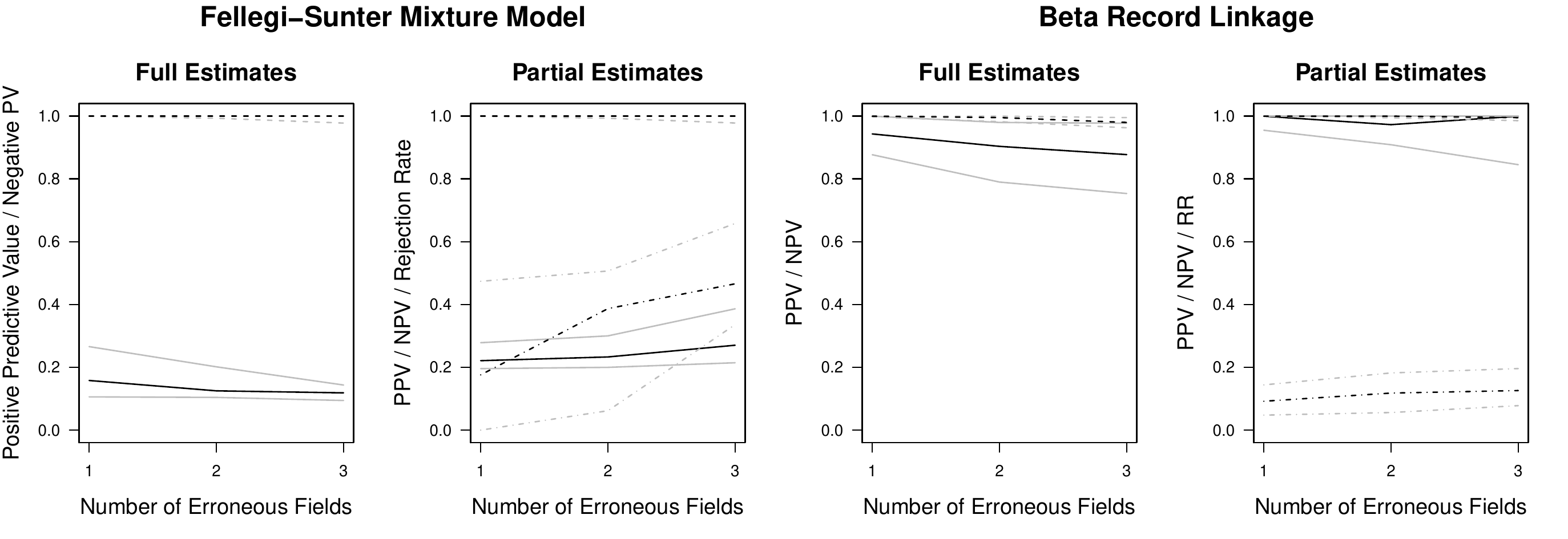}}
  \begin{minipage}[b]{0.95\textwidth}
  \caption{Performance comparison with full and partial estimates of the bipartite matching.  We use datafiles with 10\% overlap.  In the Fellegi-Sunter mixture model approach we obtain full estimates using the bipartite matching MLE and partial estimates using the Fellegi-Sunter decision rule.  In beta record linkage we use $\bZ^{e01}$ for full estimates and $\bZ^{e01R}$ for partial estimates.  Solid lines refer to precision or positive predictive value (PPV), dashed lines to negative predictive value (NPV), dot-dashed lines to rejection rate (RR), black lines show medians, and gray lines show first and 99th percentiles.  }
\label{f:BRLAustraliaSim_pointest}\end{minipage} 
\end{figure*}

As we saw in the last section, when the files have a small overlap the matching uncertainty is large and both procedures have their worst performances.  We therefore focus on the scenario where the files have 10\% overlap.  In Figure \ref{f:BRLAustraliaSim_pointest} we present the performance of both methodologies with and without rejections.  The PPV (precision) of the Fellegi-Sunter mixture model approach is very low, and it does not improve much by allowing rejections, even after using the small nominal false-match rate $\mu=0.0025$.  On the other hand, the beta approach with rejections leads to a PPV much closer to 1 across all measurement error levels, and it has a lower rejection rate.  We notice that although using the rejection option helps to prevent the PPV from being too low, as there is more error in the files the uncertainty in the matching decisions increases, and so does the rejection rate.

The simulation results presented in this section make us confident that beta record linkage represents a reliable approach for merging datafiles, especially when these contain a limited amount of information and a small overlap, such as the case that we now study.

\section{Combining Lists of Civilian Casualties from the Salvadoran Civil War}\label{s:SV}

The Central American republic of El Salvador underwent a civil war from 1980 until 1991. Over the course of the war a number of organizations collected reports on human rights violations, in particular on civilian killings.  It is unreasonable to assume that any organization covered the whole universe of violations, but their information can be combined to obtain
a more complete account of the lethal violence.  When combining such sources of information it is essential to identify which individuals appear recorded across the multiple databases.  In particular, this is a crucial step required to produce estimates on the total number of casualties using capture-recapture or multiple systems estimation \citep[see, e.g.,][]{LumPriceBanks13}, or to at least provide a lower bound on that number \citep{Jewelletal13}.  In this section we apply the beta record linkage methodology to combine the lists of civilian casualties obtained from two different sources: El Rescate - Tutela Legal (ER-TL) and the Salvadoran Human Rights Commission (Comisi\'on de Derechos Humanos de El Salvador --- CDHES).

The Los Angeles-based NGO \emph{El Rescate} developed a database of reports on human rights abuses linked to  the command structure of the perpetrators \citep{Howland08}.  The information on human rights abuses was digitized from reports that had been published periodically during the civil war by the project \emph{Tutela Legal}  of the Archdiocese of San Salvador.  Tutela Legal's information was obtained from individuals who came to their office in San Salvador to make denunciations.  Personnel from Tutela Legal interviewed the complainants, checked the credibility of the testimonies, and also compared the denunciations with their existing records to avoid duplication.   According to \cite{Howland08}, Tutela Legal required investigating all denunciations before publishing them as human rights violations, which gives us confidence on the quality of the information of this datafile.  Nevertheless, these investigations were not carried out when there were military operations or other restrictions in the area, therefore leading to many denunciations not being published, which in turn implies an important undercount of violations in this datafile. 

The second datafile comes from the CDHES.  According to \cite{Ball00ElSalvador}, between the years 1979 and 1991, the CDHES took more than 9,000 testimonials on human rights violations that were recorded and stored in written form.  \cite{Ball00ElSalvador} describes the 1992 project of digitization of all these reports, as well as the construction of a database containing summary information of the violations.   We notice that this database was constructed from testimonials that were provided to the CDHES shortly after the violations occurred, and therefore we would expect the details of the events to be quite reliable, that is, the dates and locations of the events and the names of the victims should be quite accurate, although this file is not free of typographical errors.

The two datafiles have the following six fields in common: given and family names of the victim; year, month, day, and region of death.  These two datafiles contain some records corresponding to members of families that were killed the same day in the same location, and therefore their records share the same information except perhaps for the field of given name.  Given this idiosyncrasy of the datafiles and their limited amount of information we expect the linkage to be quite uncertain for some records, and therefore the methodology proposed in this article is well suited to address this scenario.  

\subsection{Implementation of Beta Record Linkage}

 In this article, a valid
casualty report is defined as a record  that specifies given and
family name of the victim, which leads to $n_1=$ 4,420 records for ER-TL and $n_2=$ 1,324 for CDHES.  The names were standardized to account for possible misspellings that can occur with Hispanic names, as presented in \cite{Sadinle14}.  The number of record pairs is 5,852,080 and for each of them we build a comparison vector using the disagreement levels presented in Table \ref{t:SV_compdata}.  We use a modification of the Levenshtein distance introduced by \cite{Sadinle14} to account for the fact that Hispanic names are likely to be have missing pieces.  

\begin{table}[t]
\centering
\footnotesize{
  \begin{minipage}[b]{.95\textwidth}
 \caption{Construction of disagreement levels for the linkage of the ER-TL and CDHES datafiles.  The modified Levenshtein distance is standardized to be in the unit interval.}\label{t:SV_compdata}
\vspace{3pt}
\centering
\begin{tabular}{lccccc}
  \hline\\[-8pt]
       &                     & \multicolumn{4}{c}{Levels of Disagreement}\\
          \cline{3-6}\\ [-6pt]
  Fields & Similarity & $0$ & $1$ & $2$ & $3$ \\
  \hline\\[-8pt]
  Given Name, Family Name & Modified Levenshtein &  0 & $(0,.25]$ & $(.25,.5]$ & $(.5,1]$ \\
  Year of Death & Absolute Difference &  0 & 1 & 2 & 3+ \\
  Month of Death & Absolute Difference & 0 & 1 & 2--3 & 4+ \\
  Day of Death & Absolute Difference &  0 & 1--2 & 3--7 & 8+ \\
	Region of Death & Adjacency & Same & Adjacent & Other & \\
  \hline
\end{tabular}
\end{minipage}
}
\end{table}

We used the same implementation of beta record linkage used in the simulation studies of Section \ref{s:simulations}, that is, we used flat priors on the $\bm_f$ and $\bu_f$ parameters for all $f$, and also on the proportion of matches $\pi$.  We ran 2,000 iterations of the Gibbs sampler presented in Section \ref{ss:Gibbs}, and discarded the first 200 as burn-in.  The runtime was of 46 minutes using the same implementation as in our simulation studies.

To check convergence we computed numeric functions of the bipartite matchings in the chain.  Given a bipartite matching labeling $\bZ^{[t]}$ at iteration $t$ we found the files' overlap size $n_{12}(\bZ^{[t]})$ and the matching statuses $I(Z^{[t]}_j=i)$ for all pairs of records.  For each of these chains we computed Geweke's convergence diagnostic as implemented in the R package \verb"coda" \citep{CODA}.  In Figure \ref{f:SV_diag} we show the values of Geweke's Z-scores for the matching statuses that are  not constant in the chain.  These scores range around the usual values of a standard normal random variable, indicating that it is reasonable to treat these chains as drawn from their stationary distributions.  We also present the traceplot of the files' overlap size, from which we can see that this chain seems to have converged rather quickly.

\begin{figure*}[t]
\centering
		\centerline{\includegraphics[width=0.9\linewidth]{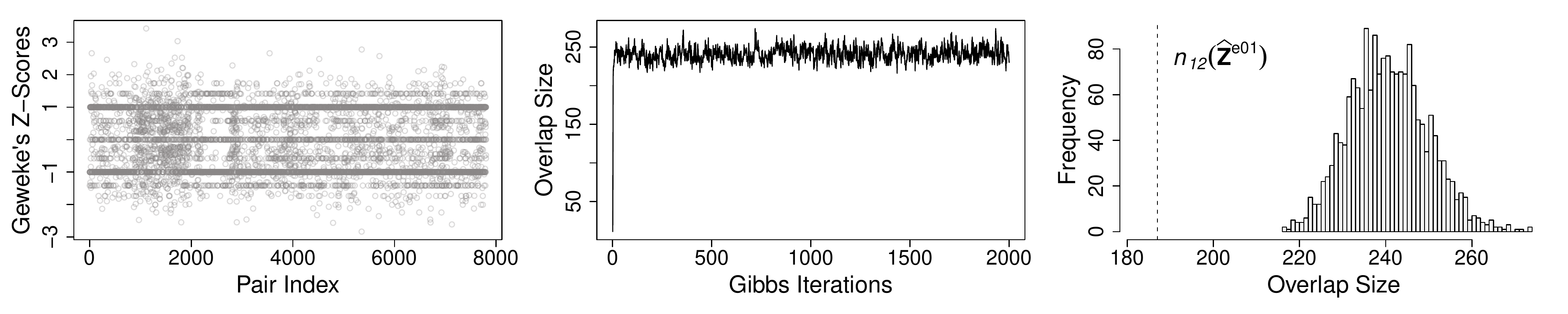}}
  \begin{minipage}[b]{0.95\textwidth}
  \caption{Left panel: Geweke's Z-scores for chains of record pairs' matching statuses. Middle panel: traceplot of the files' overlap size $n_{12}$. Right panel: estimated posterior distribution of $n_{12}$.}
\label{f:SV_diag}\end{minipage} 
\end{figure*}

\subsection{Linkage Results}

In the right-hand side panel of Figure \ref{f:SV_diag} we present the estimated posterior distribution of the files' overlap size $n_{12}$. This distribution ranges between 216 and 274, has a mean
and median of 241, and a posterior 90\% probability interval of [227, 256], with a corresponding interval of [.17, .19] for the fraction of matches $n_{12}/n_2$.  We can also obtain the posterior distribution on the number of unique killings reported to these institutions $n_1+n_2-n_{12}$, which has a 90\% probability interval of [5488, 5518].

We computed the full estimate $\hat \bZ^{e01}$, which leads to $n_{12}(\hat \bZ^{e01})=187$ links.  Comparing the posterior distribution of $n_{12}$ with $n_{12}(\hat \bZ^{e01})$ makes it evident that the $\hat \bZ^{e01}$ estimator is very conservative when declaring links.  Similarly as in our simulation studies, we also computed the partial estimate $\hat \bZ^{e01R}$ presented in Theorem \ref{prop:partial_link_closed_form} with $\lambda_{10}=\lambda_{01}=1,\lambda_{11'}=2$, and $\lambda_{R}=.1$.  Under this partial estimate the preliminary number of links is 136, and the number of rejections is 169, indicating that after clerical review the final number of links could be anywhere between 136 and 305, which contains the range of variation of the posterior of $n_{12}$.  An exploration of the rejections shows that many of them have more than one record in file 1 that could be a match, similar to the case presented in Figure \ref{f:UncertainMatch}.  For some other cases there is a single best possible match but they have a moderate number of disagreements that do not allow the pair to be linked right away.  In Table \ref{t:SV_compare_estimates} we compare the estimate $\hat \bZ^{e01R}$ against $\hat \bZ^{e01}$ and also against the bipartite matching MLE with and without the Fellegi-Sunter decision rule.

We saw that the Fellegi-Sunter mixture model approach has poor performance when linking files with a small number of fields and with a potentially small overlap, as it is the case for the files studied in this section.   For these files the bipartite matching MLE leads to a large number links (see Table \ref{t:SV_compare_estimates}) which indicates that it is probably overmatching, as we saw from our simulation studies.  When we use the Fellegi-Sunter decision rule after the MLE we obtain a partial estimator that still has a large number of links. 

\begin{table}[t]
\centering
\footnotesize{
  \begin{minipage}[b]{.95\textwidth}
 \caption{Comparison of bipartite matching estimates for the Salvadoran casualties data.  We show the cross-classifications of each type of decision (link: $\hat Z_j\leq n_1$, rejection: $\hat Z_j=R$, no-link: $\hat Z_j= n_1+j$) for the $\hat \bZ^{e01R}$ estimate versus $\hat \bZ^{e01}$, MLE, and MLE plus the Fellegi-Sunter decision rule.  The cells counting the number of records declared to have links by both estimates contain two values: $\sum_{j=1}^{n_2} I(\hat Z_j, \hat Z_j'\leq n_1) \big[\sum_{j=1}^{n_2} I(\hat Z_j=\hat Z_j' \leq n_1)\big]$, for estimates $\hat \bZ$ and $\hat \bZ'$.
}\label{t:SV_compare_estimates}
\centering
\begin{tabular}{lrrrrrrrrrrr}
  \hline\\[-8pt]
       &  && \multicolumn{2}{c}{$\hat\bZ^{e01}$} && \multicolumn{2}{c}{MLE} &&\multicolumn{3}{c}{MLE + Fellegi-Sunter}\\
          \cline{4-5} \cline{7-8} \cline{10-12}\\ [-6pt]
  $\hat \bZ^{e01R}$ & Total && \multicolumn{1}{c}{Link} & NL && \multicolumn{1}{c}{Link} & NL && \multicolumn{1}{c}{Link} & R & NL \\
  \hline\\[-8pt]
  Link            & 136 && 136 [136] &    0 && 136 [116] &   0 && 133 [116] &   3&   0  \\
	Rejection (R)  & 169  &&  51 &  118 && 167 &   2 && 157&  10&   2 \\
	Non-Link (NL)  & 1,019&&   0 & 1,019&& 861 & 158 && 466& 375& 178 \\
	\\ [-6pt]
	\multicolumn{1}{r}{Total} & 1,324 && 187 & 1,137 && 1,164 & 160 && 756 & 388 & 180\\
  \hline
\end{tabular}
\end{minipage}
}
\end{table}

The differences in the results between the Fellegi-Sunter mixture model approach and beta record linkage can be partially explained from examining the estimates of the $\bm_f$ and $\bu_f$ parameters under both approaches.  While the estimates of the $\bu_f$ are nearly the same for all fields for both methodologies, huge discrepancies appear in the $\bm_f$ parameters.  Under the Fellegi-Sunter mixture model approach $\hat\bm_{GivenName}=(.015, .001, .008, .976)$, and it is nearly the same as $\hat\bu_{GivenName}$, indicating that the given name field was essentially not taken into account to separate the class associated with matches, and therefore the clusters obtained by the mixture model are not appropriate for record linkage.  On the other hand the posterior mean of $\bm_{GivenName}$ under beta record linkage is $(.693, .050, .014, .243)$, which actually reflects something we would expect: matches tend to agree in the field given name (with 69.3\% probability).  

Although there is no ground truth for these datafiles to fully assess the performance of our methodology, the exploration of the different bipartite matching estimates and the results of our simulation studies make us confident that the beta record linkage approach along with one of the partial estimates presented in Section \ref{ss:BRLpointest} provide a good way of tackling the merging of these datafiles.  

\section{Discussion and Future Work}

The mixture model implementation of the Fellegi-Sunter approach to record linkage works well when there is not a lot of error in the files and their overlap is large.  This approach is also appealing because it is fast, but it is outperformed by the beta approach in a wide range of scenarios.  Beta record linkage provides a posterior distribution on the bipartite matchings which allows us to use different point estimators, including those that permit a rejection option with the goal withholding final linkage decisions for uncertain parts of the bipartite matching.  Although the Fellegi-Sunter decision rule was designed for this same purpose, its optimality relies on the assumption that the linkage decision for a record pair is determined only by its comparison vector once the distributions of the comparison data for matches and non-matches are fixed.  This assumption clearly does not hold true in the bipartite record linkage scenario since linkage decisions are interdependent.  We also notice that if we wanted to use the decision rules presented in Section \ref{ss:BRLpointest} with the estimated matching probabilities $\hat{\mathbb{P}}(\Delta_{ij}=1|\bg_{ij}^{obs})$ from a mixture model we would obtain conflicting decisions given that the mixture model assumes that the matching statuses of the record pairs are independent of one another.  Our Bayes estimates can however be used with any Bayesian approach to record linkage that provides a posterior distribution on bipartite matchings. 

Despite the improvements presented in this article there are a number of further extensions that can be pursued.  We focused on unsupervised record linkage, but adaptations to supervised and semi-supervised settings are also desirable.  An important extension of this methodology is to the multiple record linkage context where multiple datafiles need to be merged.  Although this problem has been addressed by extending the Fellegi-Sunter approach in \cite{SadinleFienberg13}, such generalization intrinsically inherits the difficulties emphasized in this article and it is too computationally expensive.  Finally, Bayesian approaches to record linkage hold promise in allowing the incorporation of matching uncertainty into subsequent analyses of the linked data, but formal theoretical justifications for these procedures have to be developed.   

\bibliographystyle{apalike}
\bibliography{biblio}

\end{document}